\documentclass{book}
\usepackage[contrib,lang=british]{ems-book} 

\usepackage{amssymb}

\newcommand{\be}{\begin{equation}}
\newcommand{\ee}{\end{equation}}
\newcommand{\bfig}{\begin{figure}[htb]}
\newcommand{\efig}{\end{figure}}
\newcommand{\ii}{{\textrm i}}
\newcommand{\dd}{{\textrm d}}
\newcommand{\e}[1]{\,{\textrm e}^{#1}\,}
\newcommand{\Tr}{{\operatorname{Tr\,}}}
\newcommand{\sumtwo}[2]{\sum_{\substack{#1 \\ #2}}}
\newcommand{\sumthree}[3]{\sum_{\substack{#1 \\ #2 \\ #3}}}
\def\bbone{{\mathchoice {\textrm1 \mskip-4mu \textrm l} {\textrm1 \mskip-4mu \textrm l} {\textrm1 \mskip-4.5mu \textrm l} {\textrm1 \mskip-5mu \textrm l}}}

\definecolor{light-gold}{rgb}{1,0.95,0.65}

\renewcommand*\thesection{\arabic{section}}

\numberwithin{equation}{section}
\newtheorem{theorem}{Theorem}[section]
\newtheorem{proposition}[theorem]{Proposition}
\newtheorem{lemma}[theorem]{Lemma}
\newtheorem{corollary}[theorem]{Corollary}

\newcommand{\bbC}{{\mathbb C}}
\newcommand{\bbN}{{\mathbb N}}

\newcommand{\bbR}{{\mathbb R}}
\newcommand{\bbZ}{{\mathbb Z}}
\newcommand{\caB}{{\mathcal B}}
\newcommand{\caH}{\mathcal{H}}
\newcommand{\caR}{{\mathcal R}}
\newcommand{\bsi}{{\boldsymbol i}}
\newcommand{\bst}{{\boldsymbol t}}
\newcommand{\eps}{{\varepsilon}}

\newcommand{\RR}{\mathbb{R}} 
\newcommand{\ZZ}{\mathbb{Z}} 
\newcommand{\s}{\sigma}
\newcommand{\oo}{\infty}

\newcommand{\scr}[1]{^{\scriptscriptstyle (#1)}}

\makeatletter
  \def\tagform@#1{\maketag@@@{\scriptsize{(#1)}\@@italiccorr}}
\makeatother

\begin{document}
\mainmatter

\title{Reflection positivity and infrared bounds for quantum spin systems}
\titlemark{Reflection positivity and infrared bounds for quantum spin systems}

\emsauthor{1}{Jakob E. Bj\"ornberg}{J.E.~Bj\"ornberg}
\emsauthor{2}{Daniel Ueltschi}{D.~Ueltschi}


\emsaffil{1}{Department of Mathematics,
Chalmers University of Technology and the University of Gothenburg,
Sweden \email{jakob.bjornberg@gu.se}}

\emsaffil{2}{Department of Mathematics, University of Warwick,
Coventry, CV4 7AL, United Kingdom \email{daniel@ueltschi.org}}

\classification[82B20, 82B26]{82B10}


\begin{abstract}
The method of reflection positivity and infrared bounds allows to
prove the occurrence of phase transitions in systems with continuous
symmetries. We review the method in the context of quantum spin
systems. 
\smallskip

\noindent
\textit{Dedicated to Elliott Lieb on the occasion of his 90th birthday.}
\end{abstract}

\makecontribtitle


\section{Introduction}

\subsection{The quest for understanding phase transitions}

The description of phase transitions puzzled physicists at the dawn of
Statistical Physics. The first mathematical description goes back to
1924 with Einstein's description of Bose--Einstein condensation, but
its physical relevance was not recognised then. Peierls' argument for
the Ising model was published in 1936 but  did not settle the
question either. Legend has it that in Amsterdam in 1937, during the
conference celebrating the centenary of the birth of van der Waals, a
spirited debate took place regarding the validity of the Gibbs
formalism, and of the thermodynamic limit. This was supported by Born
and Uhlenbeck but strongly contested by Sommerfeld. Kramers, who was
chairman, apparently called for the question to be settled by a
vote{\dots} and the outcome was positive, narrowly! 

The setting has been clarified over the years and given precise
mathematical meaning. The main challenge now is to prove the occurrence
of phase transitions in specific systems. 

The first results dealt with systems with discrete symmetries such as
the Ising model. Regarding systems with continuous symmetries, Mermin
and Wagner proved a negative result in 1966, namely that the
Heisenberg model does not display spontaneous magnetisation in two
dimensions, at any positive temperature \cite{MW}. It took another
decade for the first positive result to appear, due to Fr\"ohlich,
Simon, and Spencer; they established in 1976 that the classical
Heisenberg model undergoes a phase transition in dimensions three and
higher \cite{FSS}.
Their work was inspired by ideas from quantum field field theory, specifically by the K\"all\'en--Lehmann representation of two-point Green functions in relativistic quantum field theory,
which suggested the right form of infrared bounds, and by reflection positivity, as
formulated in the works of Jost \cite{Jost}, Osterwalder and Schrader \cite{OS}, and Glaser \cite{Gla}. (Furthermore, bounds in \cite{GJ,Fro1} inspired the exponential infrared bounds proved in \cite{FSS}.)


The extension of these ideas to quantum spin systems was achieved in
another groundbreaking article, by Dyson, Lieb, and Simon in 1978
\cite{DLS}. The method was then further extended and streamlined in
\cite{FILS1} and \cite{FILS2}. Further refinements include an
extension to the ground states in two dimensions \cite{NP} and
improved conditions that establish long-range order in the XY model in
two dimensions \cite{KLS1, KLS2, KK}.

It should be pointed out that the method does not apply to models where all coupling constants are positive \cite{Spe}. An important problem, which remains open to this day, is to prove spontaneous magnetisation or long-range order in the Heisenberg ferromagnet. 

Another extension of the method deals with "chessboard estimates",
proposed by Fr\"ohlich and Lieb \cite{FL} (they were partly motivated by \cite{GJS}).  Among many interesting
works that use these ideas, let us mention the flux phase problem
\cite{Lieb1, MN}; spin reflection positivity applied to Hubbard models
\cite{Lieb2, Tian, Tas}; itinerant electron models \cite{LN, Mac, ML};
high spin systems whose classical limit has long-range order
\cite{BC,BCS}; spin nematic phases \cite{AZ, TTI}; Néel order in the spin-1 model with biquadratic interactions \cite{Lees}; hard-core bosons
\cite{KLS2, ALSSY}; loop models associated with quantum spin systems \cite{Uel} (motivated by \cite{Toth1,AN}) and other loop models associated with classical spin systems \cite{QT}. Finally, let us mention an alternate extension of \cite{FSS} to quantum systems by Albert, Ferrari, Fr\"ohlich, and Schlein \cite{AFFS}.

There exist a few results about phase transitions in systems with continuous symmetry that were proved with different methods, see \cite{BHH,GS}. But reflection positivity and infrared bound remains the most prolific method and the only one that has been applied to quantum systems.

A beautiful account of the method of reflection positivity in statistical mechanics has been
written by Biskup \cite{Bis}. It is restricted to classical
systems, so the present survey can be seen as a
complement, dealing with the quantum counterparts. Nevertheless, we
have attempted to write a self-contained 
pedagogical account, that encompasses many of the results on
phase transitions for quantum
spin systems. 

The handwritten notes of T\'oth for his Prague lectures give a clear account of the method \cite{Toth2}. And 
an extensive overview, which retraces the origin of the key ideas, can be found in the handwritten notes of Fr\"ohlich for his Vienna lectures \cite{Fro2}.

\subsection{Organisation of the survey}

The setting for quantum spin systems is introduced in Section 
\ref{sec setting}.  We recall  the existence of the
infinite-volume limit of the free energy 
(Theorem \ref{thm thermo lim}), and discuss non-differentiability of
the free energy and how it relates to long-range order.
Subsequently, the chapter is organised as follows:
\begin{itemize}[leftmargin=*]
\item The existence of long-range order is stated in Section 
\ref{sec LRO}. 
For this we consider the case of positive temperature in dimensions 3
and higher as well as the ground state of
two-dimensional systems. 
\item The case of nearest-neighbour interactions is covered in Theorem \ref{thm LRO nn}. As discussed after the theorem, long-range order has been proved for all $d \geq 2$ and all $S \in \frac12 \bbN$, except for the case $d=2$ and $S=\frac12$ where it is restricted to models close to XY.
\item The proof of long-range order in turn relies on an "infrared
  bound" on correlations which we state in Section \ref{sec IRB}. This involves transferring the bound from the Duhamel correlation function to the normal correlation function.
\item We establish reflection positivity of our models in Section \ref{sec RP} and use it to prove the infrared bound on the Duhamel correlation function.
\item The appendix contains correlation inequalities for quantum spin systems that we need in various places.
\end{itemize}

\section{Setting and results}
\label{sec setting}

The domain of the system is a finite subset $\Lambda \subset
\bbZ^d$. The state space is the Hilbert space $\caH_\Lambda =
\otimes_{x\in\Lambda} \bbC^n$ where $n = 2S+1 =
2,3,4,\dots$. Equivalently, $\caH_\Lambda$ can be defined as the
linear span of the space of classical configurations, $\caH_\Lambda =
{\textrm{span}} (\Sigma^\Lambda)$, where $\Sigma = \{-S,-S+1,\dots,S\}$. We
consider the spin operators $(S_x\scr{i})_{x\in\Lambda}^{i=1,2,3}$
that satisfy the relations: 
\be
\label{spin relations}
\begin{split}
&[S_x\scr1, S_y\scr2] = \ii S_x\scr3 \delta_{x,y}, \quad [S_x\scr2, S_y\scr3] = \ii S_x\scr1 \delta_{x,y}, \quad [S_x\scr3, S_y\scr1] = \ii S_x\scr2 \delta_{x,y}, \\
& \bigl( S_x\scr1 \bigr)^2 + \bigl( S_x\scr2 \bigr)^2 + \bigl( S_x\scr3 \bigr)^2 = S(S+1) {\textrm{Id}},
\end{split}
\ee
for all $x,y\in\Lambda$. It can be shown that all these operators have eigenvalues $\{-S,\dots,S\}$. For $S=\frac12$, these operators are given by (half) the Pauli matrices in the basis where $\{S_x\scr3\}$ are diagonal, namely
\be
\begin{split}
&S_x\scr1 = \tfrac12 \bigl( \begin{smallmatrix} 0 & 1 \\ 1 & 0 \end{smallmatrix} \bigr) \otimes {\textrm{Id}}_{\Lambda \setminus \{x\}}, \\
&S_x\scr2 = \tfrac12 \bigl( \begin{smallmatrix} 0 & -\ii \\ \ii & 0 \end{smallmatrix} \bigr) \otimes {\textrm{Id}}_{\Lambda \setminus \{x\}}, \\
&S_x\scr3 = \tfrac12 \bigl( \begin{smallmatrix} 1 & 0 \\ 0 & -1 \end{smallmatrix} \bigr) \otimes {\textrm{Id}}_{\Lambda \setminus \{x\}}.
\end{split}
\ee
These expressions are adapted to the tensor product structure
$\caH_\Lambda = \otimes_{x\in\Lambda} \bbC^n$. Using the
interpretation $\caH_\Lambda = {\textrm{span}} (\Sigma^\Lambda)$ instead, we can
define the spin operators as follows. Given $\sigma_\Lambda \in \Sigma^\Lambda$, let $\sigma_\Lambda^{(x)}$
denote the configuration that is equal to $\sigma_\Lambda$, except at
$x$ where the value has been flipped; then 
\be
\begin{split}
&S_x\scr1 |\sigma_\Lambda\rangle = |\sigma_\Lambda^{(x)} \rangle, \\
&S_x\scr2 |\sigma_\Lambda\rangle = \ii \sigma_x |\sigma_\Lambda^{(x)} \rangle, \\
&S_x\scr3 |\sigma_\Lambda\rangle = \sigma_x |\sigma_\Lambda \rangle.
\end{split}
\ee
Here $|\sigma_\Lambda\rangle$ is the vector in $\caH_\Lambda$ that corresponds to $\sigma_\Lambda$.
The Hamiltonian of the system is the operator
\be
\label{def ham qu}
H_{\Lambda,h} = -\sum_{i=1}^3 \sum_{x,y\in\Lambda} J_{x-y}\scr{i} S_x\scr{i} S_y\scr{i} - h \sum_{x \in \Lambda} S_x\scr3
\ee
and the partition function is
\be
\label{def partf qu}
Z(\Lambda,\beta,h) = \Tr \e{-\beta H_{\Lambda,h}}.
\ee
The functions $J_{x-y}\scr{i}$ are called coupling parameters or
coupling constants. We always assume that they are real and symmetric, $J_x\scr{i} = J_{-x}\scr{i} \in \bbR$ for all $i \in \{1,2,3\}$ and $x \in \bbZ^d$.

\medskip

We define the finite-volume free energy by
\be
f_\Lambda(\beta,h) = -\frac1{\beta |\Lambda|} \log Z(\Lambda,\beta,h),
\qquad f_\Lambda(\oo,h)=\lim_{\beta\to\oo} f_\Lambda(\beta,h).
\ee
The limit $\beta\to\oo$ exists and it corresponds to a trace in the eigenspace
for $H_{\Lambda,h}$ with lowest eigenvalue (the ground-state energy).
As is well-known, we can take the limit of large volumes and we obtain
a thermodynamic potential. This is a pillar of Statistical
Mechanics. For this, we recall the notion of "van Hove sequences". We
say that the sequence of finite domains $(\Lambda_n)$ tends to
$\bbZ^d$ \emph{in the sense of van Hove},
written $\Lambda_n\Uparrow\ZZ^d$,  if 
\begin{itemize}
\item[(i)] it is increasing, i.e.\ $\Lambda_n \subset \Lambda_{n+1}$;
\item[(ii)] it invades $\bbZ^d$, i.e.\ $\cup_n \Lambda_n = \bbZ^d$;
\item[(iii)] its ratio boundary/volume tends to zero, 
i.e.\ $\frac{|\{ x \in \Lambda_n : 
{\textrm{dist}}(x,\Lambda_n^{\textrm c}) = 1\}|}{|\Lambda_n|} \to 0$.
\end{itemize}

\begin{theorem}
\label{thm thermo lim}
Assume that
\[
\sum_{x\in\bbZ^d} |J_x\scr{i}| < \infty
\]
for all $i \in\{1,2,3\}$. Then there exists a function 
$f(\beta,h)$ where $0\leq\beta\leq \oo$ and $h\in\bbR$, such that for
any van Hove sequence $(\Lambda_n)$, we have 
\[
f(\beta,h) = \lim_{n\to\infty} f_{\Lambda_n}(\beta,h).
\]
Convergence is locally uniform. Further, the function $f(\beta,h)$ is
concave and even in $h$. 
\end{theorem}

We refer to \cite{Rue,FV} for the proof of this theorem. 
For finite $\beta$,  the function $\beta f(\beta,h)$ is
jointly concave in $(\beta, \beta h)$ (but here we only use the
concavity in $h$).  
It is also
well-known that $f(\beta,h)$ is smooth when $\beta$ is small --- it is
actually analytic. 

In what follows we need to consider periodic boundary
conditions. Given $\ell \in \bbN$, let $\Lambda_\ell =
\{0,1,\dots,\ell-1\}^d$, and let the Hamiltonian
$H_{\Lambda_\ell,h}^{\textrm{per}}$ be given by Eq.\ \eqref{def ham qu}, but
with coupling parameters replaced by the following periodised ones: 
\be\label{eq:J-per}
J_{x,{\textrm{per}}}\scr{i} = \sum_{z \in \bbZ^d} J_{x + \ell z}\scr{i}.
\ee
We can define the periodised partition function $Z^{\textrm{per}}(\Lambda_\ell, \beta, h)$ accordingly, and the free energy
\be
f^{\textrm{per}}_{\Lambda_\ell}(\beta,h) = -\frac1{\beta \ell^d} \log Z^{\textrm{per}}(\Lambda_\ell, \beta, h).
\ee
As $\ell \to \infty$, these free energies converge to the free
energies $f(\beta,h)$ of Theorem \ref{thm thermo lim}.

We introduce the finite-volume equilibrium states
\be
\langle \cdot \rangle_{\Lambda, \beta, h}=
\frac{\Tr[\cdot \e{-\beta H_{\Lambda,h}}]}
{Z_{\Lambda,\beta,h}},
\qquad
\langle\cdot\rangle_{\Lambda,\oo,h}=\lim_{\beta\to\oo}
\langle\cdot\rangle_{\Lambda,\beta,h}. 
\ee
We also consider the states $\langle \cdot \rangle^{\textrm{per}}_{\Lambda_\ell, \beta, h}$ with periodic boundary conditions, where we use $H^{\textrm{per}}_{\Lambda_\ell,h}$ instead of $H_{\Lambda_\ell,h}$.
In Section \ref{sec LRO} we make several assumptions on the coupling
constants and we show that the system then exhibits \textit{long-range
  order} at low temperature, in the sense that there exists a lower
bound $c$ such that 
\be
\label{def LRO}
\frac1{|\Lambda|^2} \sum_{x,y \in \Lambda} \langle S_x\scr3 S_y\scr3 \rangle_{\Lambda,\beta,0} \geq c > 0,
\ee
where $c$ is positive uniformly in the volume $\Lambda$. As for the
domains $\Lambda$, the statement is relevant if it holds for all
domains in a van Hove sequence. 

In order to motivate the importance of this property, we show that it implies the occurrence of a first-order phase transition as $h$ crosses 0. This also implies that there exist many distinct Gibbs
states at $(\beta,0)$.

\begin{theorem}
\label{thm non diff}
We assume that the system displays long-range order in the form of Eq.\ \eqref{def LRO}.
Then
\[
\frac\partial{\partial h} f(\beta,h) \Big|_{h=0-} \; > \; 0 
\; > \; \frac\partial{\partial h} f(\beta,h) \Big|_{h=0+}.
\]
\end{theorem}


\begin{proof}
We first give a simplified proof in the case where $[H_{\Lambda,h},M_\Lambda] = 0$, where $M_\Lambda$ is the magnetisation operator
\be
M_\Lambda=\sum_{x\in\Lambda} S\scr3_x.
\ee
For the general case, we will use a result of Koma and Tasaki \cite{KT}.

Let $|M_\Lambda|$ be the unique positive semi-definite square root of $M_\Lambda^2$.
We have $|M_\Lambda| \leq |\Lambda| S \, {\textrm{Id}}$, so that $M_\Lambda^2 \leq |\Lambda| S |M_\Lambda|$. Since Gibbs states are positive linear functionals, we get
\be
\Bigl\langle \frac{M_\Lambda^2}{|\Lambda|^2} \Bigr\rangle_{\Lambda,\beta,0} \leq S \Bigl\langle \frac{|M_\Lambda|}{|\Lambda|} \Bigr\rangle_{\Lambda,\beta,0}.
\ee
Long-range order implies that $\frac1{|\Lambda|^2} \langle M_{\Lambda}^2 \rangle_{\Lambda,\beta,0} \geq c$, so the right side above is positive.

In order to get an inequality for the derivative of the free energy, let us introduce $\tilde f_\Lambda(\beta,h)$ to be the free energy of the model with Hamiltonian
\be
\tilde H_{\Lambda,h} = -\sum_{i=1}^3 \sum_{x,y\in\Lambda}
J_{x-y}\scr{i} S_x\scr{i} S_y\scr{i} - 
h |M_\Lambda|.
\ee
We now check that $\tilde f_\Lambda(\beta,h)$
converges as $\Lambda \Uparrow \bbZ^d$ to the free energy $f(\beta,h)$
for $h\geq0$. For this, notice that $M_\Lambda$ (and $|M_\Lambda|$) commute with $H_{\Lambda,0} = \tilde H_{\Lambda,0}$. For $h\geq 0$ we have the inequalities (for the second one, observe that the spectrum of $M_\Lambda$ is symmetric around 0)
\be
\Tr \e{-\beta H_{\Lambda,\beta,0} + \beta h M_\Lambda} \leq \Tr \e{-\beta H_{\Lambda,\beta,0} + \beta h |M_\Lambda|} \leq 2 \, \Tr \e{-\beta H_{\Lambda,\beta,0} + \beta h M_\Lambda}.
\ee
Taking the logarithm and dividing by $\beta|\Lambda|$, and
taking the relevant limits, we get that $f$ and $\tilde f$ are equal.

We now use the concavity in $h$ of $\tilde f_\Lambda$ and the fact
that 
$\sup_n (\liminf_m a_{m,n}) \leq \liminf_m (\sup_n a_{m,n})$ and we get
\be
\label{chouettes inegalites}
\begin{split}
\frac\partial{\partial h} f(\beta,h) \Big|_{h=0+} \!\!  &= \sup_{h>0} \frac{f(\beta,h) - f(\beta,0)}h = \sup_{h>0} \liminf_{\Lambda \Uparrow \bbZ^d} \frac{\tilde f_\Lambda(\beta,h) - \tilde f_\Lambda(\beta,0)}h \\
&\leq \liminf_{\Lambda \Uparrow \bbZ^d} \sup_{h>0} \frac{\tilde f_\Lambda(\beta,h) - \tilde f_\Lambda(\beta,0)}h = \liminf_{\Lambda \Uparrow \bbZ^d} \frac\partial{\partial h} \tilde f_\Lambda(\beta,h) \Big|_{h=0} \\
&= \liminf_{\Lambda \Uparrow \bbZ^d} \Bigl\langle -\frac{|M_\Lambda|}{|\Lambda|} \Bigr\rangle_{\Lambda,\beta,0}.
\end{split}
\ee
The last expectation is with respect to the Gibbs state with Hamiltonian $\tilde H_{\Lambda,0} = H_{\Lambda,0}$. 
The right side
is positive and $\frac\partial{\partial h} f(\beta,h)
\big|_{h=0+}$ is indeed negative. Since $f$ is even in $h$ we get the other inequality as well.

When $M_\Lambda$ does not commute with the Hamiltonian, it does not seem possible to show that $f$ and $\tilde f$ are equal. But since right-derivatives of concave functions are right-continuous, we can proceed as above and get
\be
\frac\partial{\partial h} f(\beta,h) \Big|_{h=0+} = \lim_{h' \to 0+} \frac\partial{\partial h} f(\beta,h) \Big|_{h=h'+} \leq \lim_{h' \to 0+} \liminf_{\Lambda \Uparrow \bbZ^d} \Bigl\langle -\frac{M_\Lambda}{|\Lambda|} \Bigr\rangle_{\Lambda,\beta,h'}.
\ee
Koma and Tasaki \cite{KT} have proved that long-range order (in the sense of \eqref{def LRO}) implies that the right side is strictly negative.
\end{proof}

We now discuss spin rotations. They show that Hamiltonians with different couplings are related by a unitary transformation, which allows to make assumptions on the couplings without loss of generality.
The following lemma applies to spin operators in $\bbC^{2S+1}$, and
immediately extends to tensor products. 
Given $\vec a \in \bbR^{3}$, let
\be
S\scr{\vec a} = \vec a \cdot \vec S = a_{1} S\scr1 + a_{2} S\scr2 + a_{3} S\scr3.
\ee
By linearity, the commutation relations \eqref{spin relations} generalise as
\be
\label{gen Pauli commutations}
[S\scr{\vec a}, S\scr{\vec b}] = \ii S\scr{\vec a \times \vec b}.
\ee
Finally, let $R_{\vec a} \vec b$ denote the vector $\vec b$ rotated around $\vec a$ by the angle $\|\vec a\|$.

\begin{lemma}
\label{lem spin rotations}
\[
\e{-\ii S\scr{\vec a}} S\scr{\vec b} \e{\ii S\scr{\vec a}} = S\scr{R_{\vec a} \vec b}.
\]
\end{lemma}

\begin{proof}
We replace $\vec a$ by $s \vec a$, and we check that both sides of the identity satisfy the same differential equation. We find
\be
\frac{\dd}{\dd s} \e{-\ii S^{(s \vec a)}} S\scr{\vec b} \e{\ii S^{(s \vec a)}} = -\ii [ S\scr{\vec a}, \e{-\ii S^{(s \vec a)}} S\scr{\vec b} \e{\ii S^{(s \vec a)}} ],
\ee
and
\be
\frac{\dd}{\dd s} S\scr{R_{s \vec a} \vec b} = \Bigl( \frac\dd{\dd s} R_{s \vec a} \vec b \Bigr) \cdot \vec S = \Bigl(  \vec a \times R_{s \vec a} \vec b \Bigr) \cdot \vec S = -\ii [ S\scr{\vec a}, S\scr{R_{s \vec a} \vec b} ].
\ee
We used \eqref{gen Pauli commutations} for the last identity.
\end{proof}

We obtain the following consequence.
Here $|z|$ denotes the sum of the coordinates of $z\in\mathbb Z^d$.

\begin{proposition}\label{prop corr rotations}
Let $\rho(1),\rho(2),\rho(3)$ be a permutation of $1,2,3$ and let
$\s_1,\s_2,\s_3\in\{-1,+1\}$ satisfy $\s_1\s_2\s_3=+1$.  Let 
$H_{\Lambda,h}$ be as in \eqref{def ham qu} and let
\be
\tilde H_{\Lambda,h} = -\sum_{i=1}^3 \sum_{x,y\in\Lambda} 
\tilde J_{x-y}\scr{i} S_x\scr{i} S_y\scr{i}
-\sum_{x\in\Lambda} \tilde h_x S\scr{\rho^{-1}(3)}_x
\ee
where
$\tilde J_x\scr{i}=(\s_{\rho(i)})^{|x|} J_x\scr{\rho(i)}$
and $\tilde h_x=(\s_{\rho^{-1}(3)})^{|x|} h$.
Then for any finite $\Lambda\subset \ZZ^d$
\be
\Tr \e{-\beta H_{\Lambda,h}}=
\Tr \e{-\beta \tilde H_{\Lambda,h}}.
\ee
\end{proposition}

Using Proposition \ref{prop corr rotations}
one may transfer results on non-differentiability of the
free energy or long-range order from one set of coupling parameters to another.



\begin{proof}
For any unitary matrix $U$ acting on $\mathcal H_\Lambda$, we have
\be
\Tr \e{-\beta H_{\Lambda,h}}=
\Tr[U^\ast \e{-\beta  H_{\Lambda,h}}U]=
\Tr[\e{-\beta U^\ast H_{\Lambda,h}U}].
\ee
Consider $U$ of the form
\be
U=\bigotimes_{\substack{x\in\Lambda\\|x| \text{ odd}}} 
W_x \bigotimes_{x\in\Lambda} V_x.
\ee
By combining rotations by an angle $\tfrac\pi2$, we can choose $V_x$
such that $V_x S_x\scr{\rho^{-1}(3)}V_x^\ast=S_x\scr{3}$
while $V_x S_x\scr{\rho^{-1}(i)}V_x^\ast=\pm S_x\scr{i}$
for $i=1,2$, where the $\pm$ does not depend on $x$.
By combining rotations by an angle $\pi$,
we can choose $W_x$ such that
$W_x S_x\scr{i} W^\ast_x=\sigma_i S_x\scr{i}$.
Then $U^\ast H_{\Lambda,h}U=\tilde H_{\Lambda,h}$.
\end{proof}

\section{Long-range order}
\label{sec LRO}

We state two results about long-range order. The first theorem holds
for a larger class of coupling constants and for $S$ large enough. The
second theorem is restricted to nearest-neighbour interactions, but it
has the advantage of applying to more values of $S$ and more dimensions. 
To briefly summarise the consequences of those results, 
we will see that long-range order  (in
the form \eqref{def LRO}) holds under the following conditions:
\begin{itemize}[leftmargin=*]
\item for certain long-range interactions (specified below)
if $\beta\geq\beta_0$ for some $\beta_0<\oo$
provided $d\geq 3$ and $S$ is large enough, or
\item for nearest-neighbour interactions
if $\beta\geq\beta_0$ for some $\beta_0<\oo$
provided $d\geq 3$ and $S\geq\frac12$, or
\item for nearest-neighbour interactions
in the ground-state $\beta=\oo$
provided $d\geq 2$ and either  $S\geq1$, or
$S\geq \tfrac12$ and $-J\scr2/J\scr1\leq 0.13$. 
\end{itemize}

The original results dealt with nearest-neighbour interactions
\cite{FSS,DLS}. Then Fr\"ohlich, Israel, Lieb, and Simon formulated a
more abstract notion of reflection positive interactions \cite{FILS1}. Here we consider two cases that fit the setting;
we provide explicit proofs for each, see Lemma \ref{lem RP fpart}. 
Explicitly, we consider interactions of the following forms:
\begin{itemize}[leftmargin=*]
\item Nearest neighbour:  $J_x\scr{i}=0$ unless $\|x\|_1=1$, in
  which case it equals some constant $J\scr{i}$;
\item Fourier transform:
$J_x\scr{i}=\int_{\RR^d} \dd\nu\scr{i}(k)\e{ik\cdot x}$
where $\nu\scr{i}$ is a positive, finite measure on $\RR^d$.
\end{itemize}
The latter case allows us to include these
examples: 
\begin{itemize}[leftmargin=*]
\item $J_x\scr{i}=a\scr{i}\e{-b\scr{i}\|x\|_p^p}$
for $p\in(0,2]$ and 
constants $a\scr{i}\in\RR$, $b\scr{i}>0$.  Indeed, 
this follows from the fact that the
characteristic function of a \emph{stable distribution}
in probability theory is of the form $\e{-c|t|^p}$.
(For $p>2$ this is not possible
as the positivity of $\nu$ would be violated.)  See e.g.\  \cite{durrett}.
\item$J_x\scr{i}=a\scr{i}\|x\|_p^{-c\scr{i}}$ 
with $p\in(0,2]$, $a\scr{i}\in\RR$ and $c\scr{i}>d$.  Indeed, 
we can take linear combinations of the interactions above
with non-negative
coefficients, and we have
\be
\int_0^\infty s^{(c-1)/p} \e{-s\|x\|_p^p} \dd s 
= C \|x\|_p^{-c}.
\ee
Here $c>d$ is required  in order for the sum defining
$J_{x,\mathrm{per}}\scr{i}$ to be 
convergent.
\item Convex combinations of the above.
\end{itemize}
Let $\Lambda_\ell^*$ denote the dual of $\Lambda_\ell$
in Fourier theory, namely 
\be
\Lambda_\ell^* = \tfrac{2\pi}\ell \bigl\{ -\tfrac\ell2 + 1, \dots, \tfrac\ell2 \bigr\}^d \subset [-\pi,\pi]^d.
\ee

\begin{theorem}
\label{thm LRO}
Assume that $J\scr{i}_x$ is one of the interactions above; we assume in addition that $\ell$ is even and that
\[
J_x\scr3 \geq J_x\scr1 
\geq -J_x\scr2 \geq 0,
\qquad\mbox{for all } x\in\ZZ^d.
\]
Then
\be\label{eq:thm LRO}
\frac1{\ell^{d}} \sum_{x \in \Lambda_\ell} \langle S_0\scr3
S_x\scr3 \rangle^{\textrm{per}}_{\Lambda_\ell, \beta, 0} \geq \tfrac13
S(S+1) - \frac1{\ell^d} \sum_{k \in \Lambda_\ell^* \setminus \{0\}}
\sqrt{\frac{e(k)}{2\eps(k)}} - 
\frac1{2\beta \ell^d} \sum_{k \in \Lambda_\ell^* \setminus \{0\}} \frac1{\eps(k)}.
\ee
\end{theorem}

Here we defined
\be
\label{def eps(k)}
\eps(k) = \sum_{x \in \bbZ^d} J_{x,{\textrm{per}}}\scr3 (1 - \cos kx)
\ee
while the function $e(k)$ is defined in
 \eqref{def e(k)}. 
Notice that $\eps(k)$ is bounded and that $\eps(k) \sim k^2$ around
$k=0$; it is positive for $k \neq 0$. 
It is worth pointing out 
that $e(k) \leq {\textrm{const}}\, S^2$ around $k=0$. Therefore 
the right-hand-side of \eqref{eq:thm LRO} is
necessarily positive when $d\geq3$ and $S,\beta$ are large enough.


We now assume that $J\scr{i}$ are nearest-neighbour couplings, that is,
\be
J_x\scr{i} = \begin{cases} J\scr{i} & \text{if } \|x\|=1, \\ 0 & 
\text{otherwise.} \end{cases}
\ee
We further normalise them so that $J\scr3=1$.
In this case we derive sharper lower bounds for long-range order. Let
us introduce the following two sums: 
\be
\begin{split}
&I_\ell\scr{d} =  \frac1{\ell^d} \sum_{k \in \Lambda_\ell^* \setminus \{0\}} \sqrt{\frac{\eps(k+\pi)}{\eps(k)}}, \\
&\tilde I_\ell\scr{d} = \frac1{\ell^d} \sum_{k \in \Lambda_\ell^* \setminus \{0\}} \sqrt{\frac{\eps(k+\pi)}{\eps(k)}} \Bigl( \frac1d \sum_{i=1}^{d} \cos k_{i} \Bigr)_{+}.
\end{split}
\ee
Here, $\eps(k)=2\sum_{i=1}^d(1-\cos k_i)$ and
 $\eps(k+\pi) = 2 \sum_{i=1}^{d} (1+\cos k_{i})$, and $(\cdot)_{+}$ denotes the positive part. Their infinite volume limits converge to the integrals
\be
\label{I and I}
\begin{split}
&I\scr{d} = \lim_{\ell \to \infty} I_\ell\scr{d} =  \frac1{(2\pi)^{d}} \int_{[-\pi,\pi]^{d}} \sqrt{\frac{\eps(k+\pi)}{\eps(k)}} \dd k, \\
&\tilde I\scr{d} = \lim_{\ell \to \infty} \tilde I_\ell\scr{d} = \frac1{(2\pi)^{d}} \int_{[-\pi,\pi]^{d}} \sqrt{\frac{\eps(k+\pi)}{\eps(k)}} \Bigl( \frac1d \sum_{i=1}^{d} \cos k_{i} \Bigr)_{+} \dd k.
\end{split}
\ee
One can check that, as $d\to\infty$, these integrals satisfy $I\scr{d} \to 1$ \cite{DLS} and $\tilde I\scr{d}\to1$ \cite{KLS2}. We also introduce the expression
\be
\alpha_\ell(\beta) = J\scr1 \langle S_0\scr1 S_{e_1}\scr1 \rangle_{\Lambda_\ell,\beta,0} + J\scr2 \langle S_0\scr2 S_{e_1}\scr2 \rangle_{\Lambda_\ell,\beta,0}
\ee
and $\alpha(\beta) = \liminf_{\ell\to\infty} \alpha_\ell(\beta)$. We also denote by $\alpha_\ell(\infty)$ the $\beta\to\infty$ limit.

\begin{theorem}
\label{thm LRO nn}
Assume that $\ell$ is even and that the nearest-neighbour coupling constants satisfy
\[
J\scr3 = 1 \geq J\scr1 \geq -J\scr2 \geq 0.
\]
Then we have the two lower bounds:
\[
\begin{split}
\frac1{\ell^{d}} \sum_{x \in \Lambda_\ell} & \langle S_0\scr3
S_x\scr3 \rangle^{\textrm{per}}_{\Lambda_\ell, \beta, 0} \geq \\
&\begin{cases} 
\tfrac13 S(S+1) - \tfrac12 (I_\ell\scr{d} +\tfrac{\sqrt2}{\ell^d})
\sqrt{\alpha_\ell(\beta)} -
\frac1{2\beta \ell^d} \sum_{k \in \Lambda_\ell^* \setminus \{0\}}
\frac1{\eps(k)}, \\ \sqrt{\alpha_\ell(\beta)} \bigl[
\frac{\sqrt{\alpha_\ell(\beta)}}{1 - J\scr2 / J\scr1} - \frac12 \tilde
I_\ell\scr{d} \bigr] - \frac1{2\beta \ell^d} \sum_{k \in
  \Lambda_\ell^* \setminus \{0\}} \frac1{\eps(k)} \bigl( \frac1d
\sum_{i=1}^d \cos k_i \bigr)_+.
\end{cases}
\end{split}
\]
\end{theorem}

The theorem is proved at the end of Section \ref{sec IRB}.

We want to formulate sufficient conditions under which at least one of
the lower bounds is positive, uniformly in $\ell$. The terms involving
$1/\beta$ converge as $\ell\to\infty$ if $d \geq 3$ and they can be
made arbitrarily small by taking $\beta$ sufficiently large. For $d=2$ the bounds are useful in
the ground state, i.e.\ when the limit $\beta\to\infty$ is taken
before $\ell\to\infty$. 

We get a uniform lower bound if either
\[
\tfrac13 S(S+1) > \tfrac12 I\scr{d} \sqrt{\alpha(\beta)} \qquad \text{or} \qquad \frac{\sqrt{\alpha(\beta)}}{1 - J\scr2 / J\scr1} > \tfrac12 \tilde I\scr{d}.
\]
Irrespective of the value of $\alpha(\beta)$, at least one of the lower bound is positive if
\be
\label{the condition}
\frac{\frac13 S(S+1)}{\frac12 I\scr{d}} > \tfrac12 \tilde I\scr{d} (1 - J\scr2 / J\scr1) \quad \Longleftrightarrow \quad 1 - J\scr2 / J\scr1 < \frac{\frac43 S(S+1)}{I\scr{d} \tilde I\scr{d}}.
\ee
Values of $I\scr{d}$ and $\tilde I\scr{d}$ can be found numerically; they are listed in Table \ref{table} for $d = 2,3,4$. This allows us to verify that the condition \eqref{the condition} holds for all values of $J\scr1, J\scr2$ such that $J\scr1 \geq -J\scr2 \geq0$, all dimensions $d \geq 2$, and all spin values $S \in \frac12 \bbN$, with the \textit{one exception} of the case $d=2$ and $S = \frac12$. In this case, \eqref{the condition} holds when $-J\scr2 / J\scr1 \in [0,0.109]$.

\begin{table}[htb]
\centering
\colorbox{light-gold}{
\begin{tabular}{c|cc}
$d$ & $I\scr{d}$ & $\tilde I\scr{d}$ \\ 
\hline
2 & 1.393 & 0.6468 \\
3 & 1.157 & 0.3499 \\
4 & 1.094 & 0.2540 \\
\end{tabular}
}
\medskip
\caption{Numerical values of the integrals $I\scr{d}$ and $\tilde
  I\scr{d}$ defined in \eqref{I and I}.}
\label{table}
\end{table}

Kubo and Kishi \cite{KK} improved the interval to $[0,0.13]$ and this is the current best result. To do this, they use the variational principle with the constant state $\otimes_{x \in \Lambda_\ell} |\frac12\rangle$ to get a bound on the ground state energy. Combined with the correlation inequalities stated in Lemma \ref{lem corr ineq}, they get a lower bound for $\alpha(\infty) = \lim_{\beta\to\infty} \alpha(\beta)$, namely
\be
\alpha(\infty) \geq \frac{1/4}{2 - J\scr2 / J\scr1}.
\ee
(In \cite{KK} they consider the case $J\scr1=J\scr3 = 1$ but it is
easily extended.) This implies that the second bound of Theorem
\ref{thm LRO nn} is positive in the interval $[0,0.13]$. 

Finally, let us remark that Theorem \ref{thm LRO nn} implies N\'eel long-range order for the Heisenberg antiferromagnet. Indeed, we can rotate the spin operators on a sublattice, by an angle $\pi$  around the 2nd axis of spins. See Proposition \ref{prop corr rotations}. Then all coupling parameters are negative.

\section{Infrared bounds}
\label{sec IRB}

This section explores estimates of the Fourier
transform of correlations and their 
 consequences.   Such estimates are particularly relevant at small
Fourier parameters; this corresponds to large wavelengths, i.e.\ the
infrared spectrum for light, hence the name given by physicists. 

We need to introduce the conventions about the Fourier transform used
in this survey. 
Recall that $\Lambda_\ell^* = \tfrac{2\pi}\ell \bigl\{ -\tfrac\ell2 + 1, \dots, \tfrac\ell2 \bigr\}^d$.
The Fourier transform of a function $f : \Lambda_\ell \to \bbC$ is
\be
\widehat f(k) = \sum_{x \in \Lambda_\ell} \e{-\ii k x} f(x), \qquad k
\in \Lambda_\ell^*,
\ee
where we write $kx$ for the usual inner product $\sum_{i=1}^d k_i x_i$.
One can check that the inverse relation is then
\be
f(x) = \frac1{\ell^d} \sum_{k \in \Lambda_\ell^*} \e{\ii k x} \widehat f(k).
\ee
Note  that $\eps(k) = \widehat J\scr3(0) - \widehat J\scr3(k)$.

The first infrared bound involves the Duhamel correlation function $\eta(x)$, defined by
\be
\label{def Duhamel corr}
\eta(x) = \frac1\beta \frac1{Z_{\textrm{per}}(\Lambda_\ell,\beta,h)} \int_0^\beta \dd s 
\;\Tr S_0\scr3 \e{-s H^{\textrm{per}}_{\Lambda,h}} S_x\scr3 \e{-(\beta-s) H^{\textrm{per}}_{\Lambda,h}}.
\ee
The method of reflection positivity allows us to establish the following
infrared bound.

\begin{lemma}
\label{lem IRB Duhamel}
Let $h=0$ and $\ell$ be even. Assume that the coupling constants $J\scr{i}$ satisfy the
assumptions of Theorem \ref{thm LRO}.
Then 
\[
\widehat\eta(k) \leq \frac1{2\beta\eps(k)},
\qquad \mbox{ for all } k \in \Lambda_\ell^* \setminus \{0\}.
\]
\end{lemma}

The proof of this lemma can be found at the end of Section \ref{sec RP}.

\subsection{Falk--Bruch inequality}

We cannot use the infrared bound directly on the Duhamel function because of a lack of suitable lower bound for $\eta(0)$. The way out is to derive another bound on the ordinary correlation function. This can be done using the Falk--Bruch inequality, which was proposed independently in \cite{FB} and \cite{DLS}.

Let $\caH$ be a separable Hilbert space, $H$ a bounded hermitian operator such that $\Tr \e{-H} < \infty$, and let $\caB$ denote the space of bounded operators on $\caH$. We define the \textit{Duhamel inner product} in $\caB$ by
\be
\label{def Duh inner}
(A,B) = \frac1Z \int_0^1 \dd s \; \Tr \e{-(1-s) H} A^* \e{-sH} B, \qquad A,B \in \caB,
\ee
with $Z = \Tr \e{-H}$. We have
\be
\begin{split}
\frac\dd{\dd s} \Tr \e{-(1-s) H} A^* \e{-sH} B &= \Tr \e{-(1-s) H} [H,A^*] \e{-sH} B \\
&= \Tr \e{-(1-s) H} A^* \e{-sH} [B,H],
\end{split}
\ee
and we obtain the useful identity
\be
\bigl( [A,H], B \bigr) = \bigl( A, [B,H] \bigr).
\ee
Further,
\be
\label{FB1}
\bigl( A, [B,H] \bigr) = \frac1Z \int_0^1 \dd s \, \frac\dd{\dd s} \Tr \e{-(1-s) H} A^* \e{-sH} B = \bigl\langle [B,A^*] \bigr\rangle
\ee
where
\be
\langle \cdot \rangle = \frac1Z \Tr \cdot \e{-H}.
\ee

For a given $A \in \caB$, let us introduce the function $F(s) =\Tr \e{-(1-s) H} A^* \e{-sH} A$. We have
\be
\label{pos 2nd der}
\frac{\dd^2}{\dd s^2} F(s) = \Tr \e{-(1-s) H} [A,H]^* \e{-sH} [A,H] \geq 0
\ee
(positivity can be shown by casting the right side in the form $\Tr B^* B$). The function $F(s)$ is therefore convex. Then
\be
\label{ineq Duhamel}
\tfrac12 \langle A^* A + A A^* \rangle = \frac1{2Z} \bigl( F(0) + F(1) \bigr) \geq \frac1Z \int_0^1 F(s) \dd s = (A,A)
\ee
with equality if and only if $[A,H]=0$. The Cauchy--Schwarz inequality
of the Duhamel inner product \eqref{def Duh inner} gives 
\be
\label{conseq CS ineq}
\bigl| \bigl( A, [B,H] \bigr) \bigr|^2 \leq (A,A) \; \bigl( [B,H], [B,H] \bigr).
\ee
Using Eq.\ \eqref{FB1} to write the Duhamel inner product of
commutators as expectations in the state $\langle \cdot \rangle$, and
the inequalities \eqref{ineq Duhamel} and \eqref{conseq CS ineq} as
well as cyclicity of the trace, we get \textit{Bogolubov's inequality}
\be
\bigl| \bigl\langle [B,A^*] \bigr\rangle \bigr|^2 \leq \tfrac12 \langle A^* A + A A^* \rangle \; \bigl\langle \bigl[ [B,H], B^* \bigr] \bigr\rangle.
\ee

Inequality \eqref{ineq Duhamel} gives an upper bound for the Duhamel inner product, but we actually need a lower bound. For this, we consider the function
\be
\Phi(s) = \sqrt s \coth\tfrac1{\sqrt s}.
\ee
This function is increasing, concave, and is depicted in Fig.\ \ref{fig Phi}. One can check that
\be
\label{bounds Phi}
\sqrt s \leq \Phi(s) \leq \sqrt s + s.
\ee

\bfig
\includegraphics[width=50mm]{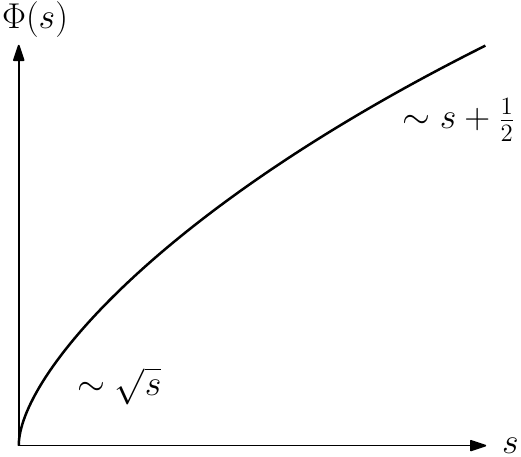}
\caption{The function $\Phi$ of the Falk--Bruch inequality.}
\label{fig Phi}
\efig

\begin{lemma}[Falk--Bruch inequality]
\label{lem FB}
For all $A\in\caB$ such that the denominators differ from zero, we have
\[
\frac{2 \langle A^* A + A A^* \rangle}{\bigl\langle \bigl[ A^*, [H,A] \bigr] \bigr\rangle} \leq \Phi \biggl( \frac{4(A,A)}{\bigl\langle \bigl[ A^*, [H,A] \bigr] \bigr\rangle} \biggr).
\]
\end{lemma}

It is worth noting that the double commutator is nonnegative, as can be seen from Eq.\ \eqref{FB1}. Indeed, taking $A \mapsto [A^*,H]$ and $B \mapsto A^*$, we can express it using the Duhamel inner product as
\be
\label{double comm}
\bigl\langle \bigl[ A^*, [H,A] \bigr] \bigr\rangle = \bigl( [A^*,H], [A^*,H] \bigr) \geq 0.
\ee

\begin{proof}[Proof of Lemma \ref{lem FB}]
Recall the function $F(s)$ defined before \eqref{pos 2nd der}. The Falk--Bruch inequality can be written as
\be
\label{FB ineq too}
2 \frac{F(0)+F(1)}{F'(1)-F'(0)} \leq \Phi \biggl( \frac{4\int_0^1 F(s) \dd s}{F'(1)-F'(0)} \biggr).
\ee
If $\{\varphi_j\}$ is an orthonormal set of eigenvectors of $H$ with eigenvalues $\lambda_j$, we can write
\be
F(s) = \sum_{i,j} \bigl| (\varphi_i, A \varphi_j) \bigr|^2 \e{-\lambda_j} \e{(\lambda_j-\lambda_i)s} = \int_{-\infty}^\infty \e{st} \dd\mu(t),
\ee
where $\mu$ is a positive measure. We have
\be
\begin{split}
&F(0) + F(1) = \int (\e{t}+1) \dd\mu(t), \\
&F'(1) - F'(0) = \int t (\e{t}-1) \dd\mu(t), \\
&\int_0^1 F(s) \dd s = \int \frac{\e{t}-1}t \dd\mu(t).
\end{split}
\ee
Let us consider the probability measure $\dd\nu(t) = ( \int t(\e{t}-1) \dd\mu(t) )^{-1} t (\e{t}-1) \dd\mu(t)$. We have
\be
\begin{split}
&\frac{F(0)+F(1)}{F'(1)-F'(0)} = \int \tfrac1t \coth \tfrac t2 \; \dd\nu(t), \\
&\frac{\int F(s) \dd s}{F'(1)-F'(0)} = \int \tfrac1{t^2} \, \dd\nu(t).
\end{split}
\ee
Since $\Phi$ is concave we can use Jensen's inequality and we get \eqref{FB ineq too}:
\be
\begin{split}
\Phi \biggl( \frac{4\int_0^1 F(s) \dd s}{F'(1)-F'(0)} \biggr) &= \Phi
\Bigl( \int \tfrac4{t^2} \, \dd\nu(t) \Bigr) \geq \int \Phi \bigl(
\tfrac4{t^2} \bigr) \dd\nu(t) \\ 
&= \int \tfrac2t \coth \tfrac t2 \; \dd\nu(t) = 2 \frac{F(0)+F(1)}{F'(1)-F'(0)}.
\end{split}
\ee
\end{proof}

The Falk--Bruch inequality is saturated when the measure $\dd\mu$ is a
Dirac on a single value. This is the case if $H$ is the Hamiltonian of
the harmonic oscillator, and $A$ is the creation or annihilation
operator. 

The following inequality follows from Lemma \ref{lem FB} and the upper bound in Eq.\ \eqref{bounds Phi}.

\begin{corollary}
\label{cor FB}
We have
\[
\tfrac12 \langle A^* A + A A^* \rangle \leq \tfrac12 \sqrt{(A,A) \; \bigl\langle \bigl[ A^*, [H,A] \bigr] \bigr\rangle} + (A,A).
\]
\end{corollary}

For our purpose we have $H \sim \beta$ and $(A,A) \sim \frac1\beta$ with $\beta$ large, so that this inequality is quite optimal. We use it below since it is simpler.

\subsection{Infrared bound for the usual correlation function}

In the rest of this section $\ell$ and $\beta$ will be fixed, and
we drop the subscripts on $\langle\cdot\rangle^{\textrm{per}}_{\Lambda_\ell,\beta,0}$,
writing simply $\langle \cdot\rangle$.

We introduce Fourier transforms of spin operators. This allows us to write the correlation functions in the form of Corollary \ref{cor FB}. Accordingly, let
\be
\widehat S_k\scr3 = \sum_{x \in \Lambda_\ell} \e{-\ii k x} S_x\scr3, \qquad k \in \Lambda_\ell^*.
\ee
One easily checks the inverse identity
\be
S_x\scr3 = \frac1{\ell^d} \sum_{k \in \Lambda_\ell^*} \e{\ii k x} \widehat S_k\scr3, \qquad x \in \Lambda_\ell.
\ee
The Fourier transform of the usual correlation function is then equal to
\be\label{eq:FT-TP}
\begin{split}
\widehat{\langle S_0\scr3 S_x\scr3 \rangle}(k) &= \sum_{x \in \Lambda_\ell} \e{-\ii k x} \langle S_0\scr3 S_x\scr3 \rangle = \frac1{\ell^d} \sum_{x,y\in\Lambda_\ell} \e{-\ii k (x-y)} \langle S_x\scr3 S_y\scr3 \rangle \\
&= \frac1{\ell^d} \langle \widehat S_{-k}\scr3 \widehat S_k\scr3 \rangle.
\end{split}
\ee 
Notice that 
$(\widehat S_k\scr3)^* = \widehat S_{-k}\scr3$, thus
\be
\label{pos corr fct}
\widehat{\langle S_0\scr3 S_x\scr3 \rangle}(k) \geq 0.
\ee
For the Duhamel correlation function we obtain
\be
\widehat \eta(k) = \widehat{ (S_0\scr3, S_x\scr3) }(k) = \frac1{\ell^d} ( \widehat S_k\scr3, \widehat S_k\scr3 ).
\ee
(There is no $-k$ because the Duhamel inner product involves taking the adjoint.)
Let
\be
\label{def e(k)}
e(k) = \tfrac12 \sum_{x \in \bbZ^d} \Bigl( (J_{x,{\textrm{per}}}\scr1 \! - \! J_{x,{\textrm{per}}}\scr2 \cos kx) \langle S_0\scr1 S_x\scr1 \rangle + (J_{x,{\textrm{per}}}\scr2 \! - \! J_{x,{\textrm{per}}}\scr1 \cos kx) \langle S_0\scr2 S_x\scr2 \rangle \Bigr).
\ee
We will see in the proof of the next lemma that $e(k)\geq0$, as it can be
written as the expectation of a double commutator in the form of Eq.\
\eqref{double comm}.

\begin{lemma}[Infrared bound for the usual correlation function]
\label{lem IRB corr fct}
We have for all $k \in \Lambda_\ell^* \setminus \{0\}$ that
\[
\widehat{\langle S_0\scr3 S_x\scr3 \rangle}(k) \leq \sqrt{\frac{e(k)}{2\eps(k)}} + \frac1{2\beta \eps(k)}.
\]
\end{lemma}

\begin{proof}
We take $A = \widehat S_k\scr3$ and $H = \beta H^{\textrm{per}}_{\Lambda,0}$ in Corollary \ref{cor FB}. We need to calculate the double commutator. First, we have
\be
\begin{split}
[H^{\textrm{per}}_{\Lambda,0}, \widehat S_k\scr3] &= \sum_{x \in \Lambda_\ell} [H^{\textrm{per}}_{\Lambda,0}, S_x\scr3] \e{-\ii kx}\\
&= -\sum_{i=1}^3 \sum_{x,y,z \in \Lambda_\ell} \e{-\ii kx} J_{y-z,{\textrm{per}}}\scr{i} [S_y\scr{i} S_z\scr{i}, S_x\scr3] \\
&= -2\ii \sum_{x,y\in\Lambda_\ell} \e{-\ii kx} \Bigl( -J_{x-y,{\textrm{per}}}\scr1 S_x\scr2 S_y\scr1 + J_{x-y,{\textrm{per}}}\scr2 S_x\scr1 S_y\scr2 \Bigr).
\end{split}
\ee
We used the fact that operators at different sites commute, and also
that $J_x\scr{i} = J_{-x}\scr{i}$. Next, 
\be
\begin{split}
\bigl[ \widehat S_{-k}\scr3, [H^{\textrm{per}}_{\Lambda,0}, \widehat S_k\scr3] \bigr] &= -2\ii \sum_{x,y\in\Lambda_\ell} \e{-\ii k x} \Bigl[ \e{\ii kx} S_x\scr3 + \e{\ii ky} S_y\scr3, -J_{x-y,{\textrm{per}}}\scr1 S_x\scr2 S_y\scr1 \\
& \hspace{58mm} + J_{x-y,{\textrm{per}}}\scr2 S_x\scr1 S_y\scr2 \Bigr] \\
&= 2 \sum_{x,y\in\Lambda_\ell} \Bigl( \bigl( J_{x-y,{\textrm{per}}}\scr1 -
\cos(k(x-y)) J_{x-y,{\textrm{per}}}\scr2 \bigr) S_x\scr1 S_y\scr1 \\
&\hspace{15mm} + \bigl( J_{x-y,{\textrm{per}}}\scr2 - \cos(k(x-y)) J_{x-y,{\textrm{per}}}\scr1 \bigr) S_x\scr2 S_y\scr2 \Bigr).
\end{split}
\ee
Taking the expectation in the Gibbs state, we obtain
\be
\bigl\langle
\bigl[ A^*, [H,A] \bigr]
\bigr\rangle
=\bigl\langle
\bigl[ \hat S\scr3_{-k}, [\beta H^{\textrm{per}}_{\Lambda,0}, \hat S\scr3_k] \bigr]
\bigr\rangle
 = 4\beta \ell^d e(k).
\ee
We also see that $e(k)\geq0$ from Eq.\ \eqref{double comm}. Lemma \ref{lem IRB corr fct} follows from Corollary \ref{cor FB} and from the infrared bound on the Duhamel correlation function, Lemma \ref{lem IRB Duhamel}.
\end{proof}

We can now prove the occurrence of long-range order.

\begin{proof}[Proof of Theorem \ref{thm LRO}]
We have the inequality (see Lemma \ref{lem corr ineq})
\be
\langle S_0\scr3 S_0\scr3 \rangle
\geq
\tfrac13\textstyle\sum_{i=1}^3 
\langle S_0\scr{i} S_0\scr{i} \rangle = \tfrac13 S(S+1).
\ee
This is where we use that 
$J_x\scr3\geq J_x\scr1 
\geq -J_x\scr2 \geq0$.

We now use the inverse Fourier transform on the two-point correlation
function, namely 
\be
\tfrac13 S(S+1) \leq \langle S_0\scr3 S_0\scr3 \rangle =
\frac1{\ell^d} \widehat{ \langle S_0\scr3 S_x\scr3 \rangle }(0) +
\frac1{\ell^d} \sum_{k \in \Lambda_\ell^* \setminus \{0\}} \widehat{
  \langle S_0\scr3 S_x\scr3 \rangle }(k). 
\ee
Notice that the first term of the right side is equal to the long-range order parameter. Then
\be
\frac1{\ell^d} \sum_{x\in\Lambda_\ell} \langle S_0\scr3 S_x\scr3
\rangle =
\frac1{\ell^d} \widehat{ \langle S_0\scr3 S_x\scr3 \rangle }(0) \geq \tfrac13 S(S+1) - \frac1{\ell^d} \sum_{k \in \Lambda_\ell^* \setminus \{0\}} \widehat{ \langle S_0\scr3 S_x\scr3 \rangle }(k).
\ee
We can bound the last term with the help of Lemma \ref{lem IRB corr
  fct}, which gives Theorem \ref{thm LRO}. 
\end{proof}

\begin{proof}[Proof of Theorem \ref{thm LRO nn}]
With nearest-neighbour interactions the function $e(k)$ can be written as
\be
e(k) = \alpha_\ell(\beta) \sum_{i=1}^d (1 + r \cos k_i),
\ee
where
\be
r = \frac{-J\scr2 \langle S_0\scr1 S_{e_1}\scr1 \rangle
- J\scr1 \langle S_0\scr2 S_{e_1}\scr2 \rangle}{J\scr1 \langle
S_0\scr1 S_{e_1}\scr1 \rangle + J\scr2 \langle S_0\scr2 S_{e_1}\scr2 \rangle}.
\ee
Here $e_1=(1,0,\dotsc,0)\in\ZZ^d$ is the unit vector in the first direction.
It follows from the fact that $e(k) \geq 0$ for all $k$, that $r \in [-1,1]$. Let
\be
I_\ell\scr{d}(r) =  \frac1{\ell^d} 
\sum_{k \in \Lambda_\ell^* \setminus \{0,\pi\}} 
\sqrt{\frac{\sum_{i=1}^d (1 + r \cos k_i)}{\sum_{i=1}^d (1 - \cos k_i)}}
\ee
where we have omitted the term $\frac1{\ell^d}\sqrt{1-r}$ for
$k=\pi=(\pi,\pi,\dotsc,\pi)$.  Adding it back 
and bounding it by $\sqrt{2}/\ell^d$,
the lower bound is 
\be
\frac1{\ell^d} \sum_{x\in\Lambda_\ell} 
\langle S_0\scr3 S_x\scr3 \rangle \geq \tfrac13 S(S+1) - 
\tfrac12 \sqrt{\alpha_\ell(\beta)} \, (I_\ell\scr{d}(r) +\tfrac{\sqrt{2}}{\ell^d})- 
\frac1{2\beta \ell^d} \sum_{k \in \Lambda_\ell^* \setminus \{0\}} \frac1{\eps(k)}.
\ee
Observe that $I_\ell\scr{d}(r)$ is concave with respect to $r$ and
that its derivative at $r=1$ is equal to 
\be
\frac\dd{\dd r} I_\ell\scr{d}(r) \Big|_{r=1} =  \frac1{\ell^d} 
\sum_{k\in\Lambda^\ast_\ell\setminus\{0,\pi\}} 
\frac{\sum_{i=1}^d \cos k_i}{\sqrt{\sum_{i=1}^d (1 - \cos k_i) \, \sum_{i=1}^d (1 + \cos k_i)}}.
\ee
This is equal to zero, as can be seen with the change of variables $k \mapsto k + (\pi,\dots,\pi)$. Then $I_\ell\scr{d}(r) \leq I_\ell\scr{d}(1) = I_\ell\scr{d}$.
Using this with the lower bound of Theorem \ref{thm LRO}, we obtain
the first bound of Theorem \ref{thm LRO nn}. 

For the second bound, we follow \cite{KLS1} and use the inverse
Fourier transform.
In what follows, $x$ is the dummy variable summed over inside the Fourier
transform.  We have
\be
\begin{split}
\langle S_0\scr3 S_{e_1}\scr3 \rangle &=
\frac1{d\ell^d} \sum_{k \in \Lambda_\ell^*}
\sum_{i=1}^d \e{\ii k_i} \widehat{\langle S_0\scr3 S_x\scr3 \rangle}(k) \\
&= \frac1{\ell^d}  \widehat{\langle S_0\scr3 S_x\scr3 \rangle}(0) + 
\frac1{\ell^d} \sum_{k \in \Lambda_\ell^* \setminus \{0\}}  
\widehat{\langle S_0\scr3 S_x\scr3 \rangle}(k)
\Big( \frac1d \sum_{i=1}^d \cos k_i\Big).
\end{split}
\ee
We used lattice symmetries and the fact that $ \widehat{\langle
  S_0\scr3 S_x\scr3 \rangle}\geq0$, see Eq.\ \eqref{pos corr fct}.  We have 
\be
\begin{split}
\frac1{\ell^d} & \widehat{\langle S_0\scr3 S_x\scr3 \rangle}(0)
\geq \langle S_0\scr3 S_{e_1}\scr3 \rangle - \frac1{\ell^d} \sum_{k \in
  \Lambda_\ell^* \setminus \{0\}}  
\widehat{\langle S_0\scr3 S_x\scr3 \rangle}(k) \Bigl( \frac1d \sum_{i=1}^d \cos k_i \Bigr)_+ \\
&\geq \langle S_0\scr3 S_{e_1}\scr3 \rangle - 
\frac1{\ell^d} \sum_{k \in \Lambda_\ell^* \setminus \{0\}}  
\Bigl( \frac1d \sum_{i=1}^d \cos k_i \Bigr)_+ \biggl[
\sqrt{\frac{e(k)}{2\eps(k)}} + \frac1{2\beta \eps(k)} \biggr]. 
\end{split}
\ee
Proceeding with $e(k)$ as we did with the first lower bound, we get
\be
\frac1{\ell^d}  \widehat{\langle S_0\scr3 S_x\scr3 \rangle}(0) 
\geq \langle S_0\scr3 S_{e_1}\scr3 \rangle - \tfrac12 \sqrt{\alpha_\ell(\beta)} \tilde I_\ell\scr{d}(r) - \frac1{2\beta \ell^d} \sum_{k \in \Lambda_\ell^* \setminus \{0\}}  \frac1{\eps(k)},
\ee
where
\be
\tilde I_\ell\scr{d}(r) = \frac1{\ell^d} \sum_{k \in \Lambda_\ell^* \setminus \{0\}} \sqrt{\frac{\sum_{i=1}^d (1 + r \cos k_i)}{\sum_{i=1}^d (1 - \cos k_i)}} \Bigl( \frac1d \sum_{i=1}^{d} \cos k_{i} \Bigr)_{+}.
\ee
One easily checks that the derivative of $\tilde I_\ell\scr{d}(r)$ is positive, so it is smaller than $\tilde I_\ell\scr{d}(1) = \tilde I_\ell\scr{d}$. Finally, using Lemma \ref{lem corr ineq}, we have
\be
\langle S_0\scr3 S_{e_1}\scr3 \rangle
\geq \frac{\alpha_\ell(\beta)}{1 - J\scr2 / J\scr1}.
\ee
The second lower bound of Theorem \ref{thm LRO nn} follows.
\end{proof}

\section{Reflection positivity}
\label{sec RP}

Let $\caH$ be a separable Hilbert space, and 
let $\caB_{\textrm{left}}$, resp.\ $\caB_{\textrm{right}}$, denote the space of
bounded operators on $\caH\otimes\caH$ that are of the form $a \otimes
\bbone$, resp.\ $\bbone \otimes a$, for some $a \in \caB(\caH)$. Let
$\caR$ denote the automorphism of $\caB(\caH\otimes\caH)$
such that  
\be
\begin{split}
&\caR (a \otimes \bbone) = \bbone \otimes a, \\
&\caR (\bbone \otimes a) = a \otimes \bbone.
\end{split}
\ee

Let us fix an orthonormal basis $\{e_i\}$ on $\caH$, and define the complex conjugate $\overline a$ of a bounded operator $a$ by
\be
\langle e_i, \overline a e_j \rangle = \overline{\langle e_i, a e_j \rangle}.
\ee
In matrix notation, that means taking the complex conjugate of its elements, without transposing as for hermitian adjoints. The reason to use the complex conjugate is that for all $a,b \in \caB(\caH)$, we have
\be
\label{prod complex conj}
\overline{a \, b} = \overline a \; \overline b.
\ee

Here is the key inequality that is closely related to reflection
positivity. Let $I$ be an index set and $\mu$ a positive, finite
measure on $I$. We assume that $A, C_i \in \caB_{\textrm{left}}$ and
$B, D_i \in \caB_{\textrm{right}}$ for all $i \in I$.

\begin{lemma}
\label{lem RP}
We have
\[
\Bigl| \Tr \e{A + B + \int C_i D_i \dd\mu(i)} \Bigr|^2 \leq \Tr \e{A + \caR \overline A + \int C_i \caR \overline C_i \dd\mu(i)} \cdot \Tr \e{\caR \overline B + B + \int \caR \overline D_i \, D_i \dd\mu(i)}.
\]
\end{lemma}

\begin{proof}
We use the Duhamel formula in the following form. If $A,B$ are bounded operators, then
\be
\label{def Duhamel}
\e{A+B} = \sum_{n\geq0} \int_{0<t_1<\dots<t_n<1} \dd t_1 \dots \dd t_n \e{t_1 A} B \e{(t_2-t_1) A} B \dots B \e{(1-t_n) A}.
\ee
In what follows, we use the shorthands
\be
\int \dd\bsi \equiv \int\dd\mu(i_1) \dots \int\dd\mu(i_n) \qquad \text{and} \qquad \int\dd\bst \equiv \int_{0<t_1<\dots<t_n<1} \dd t_1 \dots \dd t_n.
\ee
We also write $A = a \otimes \bbone$, $B = \bbone \otimes b$, $C_i = c_i \otimes \bbone$, and $D_i = \bbone \otimes d_i$. Then
\be
\begin{split}
&\Bigl| \Tr_{\caH\otimes\caH} \e{A + B + \int C_i D_i \dd\mu(i)} \Bigr|^2 \\
&= \Bigl| \sum_{n\geq0} \int\dd\bsi \int\dd\bst\; \Tr_{\caH\otimes\caH} \e{t_1(A + B)} C_{i_1} D_{i_1} \dots C_{i_n} D_{i_n} \e{(1-t_n) (A + B)} \Bigr|^2 \\
&= \Bigl| \sum_{n\geq0} \int\dd\bsi \int\dd\bst \;\Tr_{\caH} \e{t_1 a} c_{i_1} \dots c_{i_n} \e{(1-t_n) a} \Tr_{\caH} \e{t_1 b} d_{i_1} \dots d_{i_n} \e{(1-t_n) b} \Bigr|^2 \\
&\leq \sum_{n\geq0} \int\dd\bsi \int\dd\bst \;\Tr_{\caH} \e{t_1 a} c_{i_1} \dots c_{i_n} \e{(1-t_n) a} \Tr_{\caH} \e{t_1 \overline a} \overline c_{i_1} \dots \overline c_{i_n} \e{(1-t_n) \overline a} \\
&\qquad \cdot \sum_{n\geq0} \int\dd\bsi \int\dd\bst \; \Tr_{\caH} \e{t_1 \overline b} \overline d_{i_1} \dots \overline d_{i_n} \e{(1-t_n) \overline b} \Tr_{\caH} \e{t_1 b} d_{i_1} \dots d_{i_n} \e{(1-t_n) b} \\
&= \Tr_{\caH\otimes\caH} \e{A + \caR \overline A + \int C_i \caR \overline C_i \dd\mu(i)} \cdot \Tr_{\caH\otimes\caH} \e{\caR \overline B + B + \int \caR \overline D_i \,D_i \dd\mu(i)}.
\end{split}
\ee
We used the ordinary Cauchy--Schwarz inequality for functions, here
with argument $(n,\bsi,\bst)$. The complex conjugate was written with
the help of \eqref{prod complex conj}. 
\end{proof}

We now derive the infrared bound for the Duhamel correlation function,
Lemma \ref{lem IRB Duhamel}. 
In the rest of this Section, we fix an even integer $\ell$ and
consider periodic couplings \eqref{eq:J-per}.  Recall that
$\Lambda_\ell=\{0,1,\dotsc,\ell-1\}^d$. 
Let $\Delta$ denote the discrete Laplacian from the coupling constant
$J\scr3_{{\textrm{per}}}$, which acts on a field $v = (v_x) \in \bbR^\Lambda$ as
\be
(\Delta v)_x = \sum_{y \in \Lambda_\ell} 
J\scr3_{x-y,{\textrm{per}}} (v_y - v_x).
\ee

Notice the following identity, which is a discrete version of $\int f(-\Delta g) = \int \nabla f \nabla g$ for functions:
\be
\label{discrete grad}
(u, -\Delta v) = \tfrac12 \sum_{x,y \in \Lambda_\ell} J\scr3_{x-y,{\textrm{per}}} (u_x-u_y) (v_x-v_y).
\ee
In the left side, $(\cdot,\cdot)$ stands for the usual inner product on $\bbR^{\Lambda_\ell}$, i.e.\ $(u,v) = \sum_{x\in\Lambda_\ell} u_x v_x$.
We introduce the following partition function that depends on a field
$v$: 
\be
Z(v) = \Tr \e{-\beta H(v)},
\ee
with Hamiltonian given by
\be
H(v) = H^{\textrm{per}}_{\Lambda_\ell,0} - \sum_{x \in \Lambda_\ell} h_x S_x\scr3,
\ee
where the local magnetic field is obtained from $v$ by
\be
h_x = (\Delta v)_x.
\ee
Let
\be
\label{def Z'}
\tilde Z(v) = \e{\frac14 \beta (v,\Delta v)} Z(v).
\ee
We show that $\tilde Z(v)$ is maximised by the field $v \equiv 0$, which is the key to proving Lemma \ref{lem IRB Duhamel}.

Let $\caR$ denote a reflection across a plane cutting through edges. Namely, given a direction $i = 1,\dots,d$ and a half integer $\epsilon \in \{\frac12, \frac32, \dots, \frac{\ell-1}2 \}$, let $\caR$ be the bijection $\Lambda_\ell \to \Lambda_\ell$ such that
\be
\caR x = x + 2(\epsilon-x_i) e_i.
\ee
Let
\be
\begin{split}
\Lambda_{\textrm{left}} = \{ x \in \Lambda_\ell : \epsilon - \frac\ell2 < x_i < \epsilon \}, \qquad
\Lambda_{\textrm{right}} = \{ x \in \Lambda_\ell : \epsilon < x_i < \epsilon + \frac\ell2 \}.
\end{split}
\ee
Given a field $v_1 \in \bbR^{\Lambda_{\textrm{left}}}$, let $(\caR v_1)_x = (v_1)_{\caR x} \in \bbR^{\Lambda_{\textrm{right}}}$.

\begin{lemma}
\label{lem RP fpart}
Let the couplings $J\scr{i}$ satisfy the assumptions of Theorem
\ref{thm LRO}.  Then, 
for any $v_1 \in \bbR^{\Lambda_{\textrm{left}}}$ and 
$v_2 \in \bbR^{\Lambda_{\textrm{right}}}$, we have
\[
\tilde Z(v_1,v_2)^2 \leq \tilde Z(v_1,\caR v_1) \, \tilde Z(\caR v_2, v_2).
\]
\end{lemma}

We first prove the lemma in the case of nearest-neighbour couplings;
we then consider long-range interactions.

\begin{proof}[Proof of Lemma \ref{lem RP fpart} for nearest-neighbour couplings]
We cast $\tilde Z(v_1,v_2)$ in the form of Lemma \ref{lem RP}. Using
\eqref{discrete grad}, we get 
\be\label{eq:Z'-split}
\begin{split}
\tilde Z(v) &= \Tr \exp \beta \biggl\{ \tfrac18 \sum_{x,y} J_{x-y}\scr3 (v_y-v_x)^2 + \sum_{i=1}^3 \sum_{x,y} J_{x-y}\scr{i} S_x\scr{i} S_y\scr{i} \\
&\hspace{6cm} + \sum_{x,y} J_{x-y}\scr3 S_x\scr3 (v_y-v_x) \biggr\} \\
&= \Tr \exp \beta \biggl\{ \sum_{i=1}^2 \sum_{x,y} J_{x-y}\scr{i} S_x\scr{i} S_y\scr{i} + \sum_{x,y} J_{x-y}\scr3 \bigl( S_x\scr3 + \tfrac{v_x}2 \bigr) \bigl( S_y\scr3 + \tfrac{v_y}2 \bigr) \\
&\hspace{6cm} - \widehat J\scr3(0) \sum_x \bigl( S_x\scr3 v_x + \tfrac{v_x^2}4 \bigr) \biggr\}.
\end{split}
\ee
We used $J_x\scr{i} = J_{-x}\scr{i}$. This formula holds for general
couplings and we will also use it in the long-range case (with
$J_{x,{\textrm{per}}}\scr{i}$). 
We now assume that $J_x\scr{i} = 0$ 
except when $\|x\|_1=1$, in which case it equals a constant $J\scr{i}$.
Then  the above expression has the form of Lemma \ref{lem RP} by choosing
\be\label{eq:Z'-split2}
\begin{split}
&A = \beta \sum_{x,y \in \Lambda_{\textrm{left}}} \Bigl[  \sum_{i=1}^2 J_{x-y}\scr{i} S_x\scr{i} S_y\scr{i} + J_{x-y}\scr3 \bigl( S_x\scr3 + \tfrac{v_x}2 \bigr) \bigl( S_y\scr3 + \tfrac{v_y}2 \bigr)  \Bigr] \\
&\hspace{55mm} - \widehat J\scr3(0) \sum_{x \in \Lambda_{\textrm{left}}} \bigl( S_x\scr3 v_x + \tfrac{v_x^2}4 \bigr) \\
&B = \beta \sum_{x,y \in \Lambda_{\textrm{right}}} \Bigl[  \sum_{i=1}^2 J_{x-y}\scr{i} S_x\scr{i} S_y\scr{i} + J_{x-y}\scr3 \bigl( S_x\scr3 + \tfrac{v_x}2 \bigr) \bigl( S_y\scr3 + \tfrac{v_y}2 \bigr)  \Bigr] \\
&\hspace{55mm} - \widehat J\scr3(0) \sum_{x \in \Lambda_{\textrm{right}}} \bigl( S_x\scr3 v_x + \tfrac{v_x^2}4 \bigr) \\
&\int C_i D_i \dd\mu(i) = \beta \sumthree{x \in \Lambda_{\textrm{left}}}{y \in \Lambda_{\textrm{right}}}{\|x-y\|=1} \Bigl[ J\scr1 S_x\scr1 S_y\scr1 - J\scr2 (\ii S_x\scr2) (\ii S_y\scr2) \\
&\hspace{55mm} + J\scr3 \bigl( S_x\scr3 + \tfrac{v_x}2 \bigr) \bigl( S_y\scr3 + \tfrac{v_y}2 \bigr) \Bigr].
\end{split}
\ee
In the usual basis where all $S_x\scr3$ are diagonal, we have
$\overline{S_x\scr1} = S_x\scr1$, $\overline{\ii S_x\scr2} = \ii
S_x\scr2$, $\overline{S_x\scr3} = S_x\scr3$.  Then $\overline A = A$
and $\overline B = B$. We have multiplied $S_x\scr2$ by $\ii$, so
taking the complex conjugate gives the operator back. Then $\overline
{C_i} = C_i$ and $\overline {D_i} = D_i$.  Moreover,  when
$x\in\Lambda_{\textrm{left}}$ and  $y\in\Lambda_{\textrm{right}}$ 
with $\|x-y\|=1$, the  reflection
interchanges $x$ and $y$.
In order to use Lemma \ref{lem RP} the measure $\mu$ needs to be
positive, which is guaranteed by $J\scr1, J\scr3 \geq 0$ and 
$J\scr2 \leq 0$. 
\end{proof}

An important observation is that if certain interactions can be cast
in the form above, then this can also be done with convex combinations
of these interactions. We use this property below.

\begin{proof}[Proof of Lemma \ref{lem RP fpart} for long-range couplings]
We now consider the case  when 
$J_x\scr{i}=\int_{\RR^d} \dd\nu\scr{i}(k)\e{ik\cdot x}$
where $\nu\scr{i}$ is a positive, finite measure on $\RR^d$.
We see from \eqref{eq:Z'-split} that it suffices to consider a fixed
$i\in\{1,2,3\}$ and 
to simplify the notation we
dispense with the superscript $\scr i$.
We use the decomposition \eqref{eq:Z'-split2} but with
$J_{x,\textrm{per}}$ in place of $J_x$.  
It suffices to consider the cross-term  
\be
\sumtwo{x \in \Lambda_{\textrm{left}}}{y \in \Lambda_{\textrm{right}}} 
 J_{x-y, \textrm{per}} T_x T_y 
\ee
where 
$T_x\in\{S\scr1_x,\ii S\scr2_x,S\scr3_x+\frac{v_x}{2}\}$.
We aim to write this in the form $\int_I C_i D_i\;\dd\mu(i)$ in order to
apply Lemma \ref{lem RP}.
We expand 
\be
 J_{x-y, \textrm{per}} =
\sum_{z\in\ZZ^d}   J_{x-y+\ell z}=
\sum_{z\in\ZZ^d} \int_{\RR^d} \dd\nu(k)
\e{\ii k\cdot(x-y+\ell z)}
\ee
to write 
\be
\sumtwo{x \in \Lambda_{\textrm{left}}}{y \in \Lambda_{\textrm{right}}} 
 J_{x-y, \textrm{per}} T_x T_y =
\sum_{z\in\ZZ^d} \int_{\RR^d} \dd\nu(k)
\Big(\sum_{x\in \Lambda_{\textrm{left}} } \e{\ii k\cdot(x+\ell
  z/2)} T_x\Big)
\Big(\sum_{y\in \Lambda_{\textrm{right}} } \e{-\ii k\cdot(y-\ell z/2)}
T_y\Big).
\ee
This is of the required form $\int_I C_i D_i\;\dd\mu(i)$
with index set $I=\ZZ^d\times\RR^d$, except that we need the measure
$\mu$ to be finite.  In order to achieve this, we may approximate the
sum over $z\in\ZZ^d$ by a sum over $z\in \Lambda'$ and then let 
$\Lambda'\Uparrow\ZZ^d$.  The rest of the argument follows as in the
nearest-neighbour case.
\end{proof}

\begin{corollary}\label{cor GD}
For all $v \in \bbR^{\Lambda_\ell}$, we have
$\tilde Z(v) \leq \tilde Z(0)$.
\end{corollary}

\bfig
\includegraphics[width=120mm]{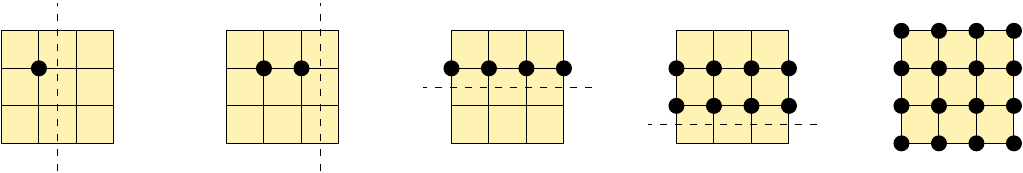}
\caption{Starting with a maximiser, reflections yield further maximisers where more and more values are identical.}
\label{fig rp maximize}
\efig

\begin{proof}
Without loss of generality we can assume that $v_0=0$. We observe that $\tilde Z(\lambda v) \to 0$ as $|\lambda| \to \infty$, so that $\tilde Z(v)$ is maximised for finite $v$. Indeed, in the expression \eqref{def Z'} we have $\e{\frac14 \beta \lambda^2 (v,\Delta v)} \sim \e{-c\lambda^2}$ and $Z(\lambda v) \leq \e{C |\lambda|}$.

Then let $(v_1,v_2)$ be a maximiser with $v_0=0$. Using Lemma \ref{lem RP fpart} with a plane crossing the edge $(0,e_1)$, we have that $(v_1,\caR v_1)$ is also a maximiser, with $v_0 = v_{e_1}=0$. Using a plane crossing the edge $(e_1,2e_1)$, we get a maximiser with more zeros. Iterating, we get a maximiser with a whole line of zeros. We then consider reflection planes in another direction to get a maximiser with a plane of zeros. We then consider reflection planes in further directions. See Fig.\ \ref{fig rp maximize} for an illustration.
\end{proof}

\begin{proof}[Proof of Lemma \ref{lem IRB Duhamel}]
From Corollary \ref{cor GD} and Eq.\ \eqref{def Z'}, we have the "Gaussian domination" bound
\be
\label{GD}
\frac{Z(s v)}{Z(0)} \leq \e{-\frac14 s^2 \beta (v,\Delta v)}.
\ee
The derivative of $Z(sv)$ with respect to $s$ is equal to 0 at $s=0$
because of symmetries (for instance, a rotation around the 3rd spin
axis by angle $\pi$, which takes $S_x\scr{i}$ to $-S_x\scr{i}$,
$i=1,2$, and leaves $S_x\scr3$ invariant). The second derivative can
be calculated e.g.\ using the Duhamel formula \eqref{def Duhamel} and
translation-invariance.  
Recalling the Duhamel correlation function $\eta$ from 
\eqref{def Duhamel corr}, we get 
\be
\frac1{Z(0)} \frac{\dd^2}{\dd s^2} Z(sv) \bigg|_{s=0} = 
\beta^2 \sum_{x,y \in \Lambda} h_x h_y \eta(x-y),
\ee
where we recall that $h_x=(\Delta v)_x$.  
We now choose the field $v$ to be
\be
v_x =  \cos(kx), \qquad  k \in \Lambda_\ell^*.
\ee
Observe that $\Delta v_x = -\eps(k) v_x$. The order $s^2$ of the inequality \eqref{GD} gives
\be
\label{GD lowest order}
\tfrac12 \beta^2 \eps(k)^2 \sum_{x,y \in \Lambda_\ell} \cos(kx) \cos(ky) \eta(x-y) \leq \tfrac14 \beta \eps(k) \sum_{x\in\Lambda_\ell} \cos(kx)^2.
\ee
Since $\eta(x)$ and $\widehat\eta(k)$ are both real, the
left-hand-side satisfies 
\be
\begin{split}
\sum_{x,y \in \Lambda_\ell} \cos(kx) \cos(ky) \eta(x-y) &= 
\sum_{x\in\Lambda_\ell} \cos(kx) \sum_{y\in\Lambda_\ell} 
\e{\ii ky} \eta(x-y) \\
&= \sum_{x\in\Lambda_\ell} \cos(kx) 
\sum_{z\in\Lambda_\ell} \e{\ii k(x-z)} \eta(z) \\
&= \sum_{x\in\Lambda_\ell} \cos(kx) \e{\ii kx} \widehat\eta(k) \\
&= \widehat\eta(k) \sum_{x\in\Lambda_\ell} \cos(kx)^2.
\end{split}
\ee
Inserting this in Eq.\ \eqref{GD lowest order} we obtain Lemma \ref{lem IRB Duhamel}.
\end{proof}

\renewcommand{\thesection}{Appendix A}


\section{Correlation inequalities for quantum systems}

\renewcommand{\thesection}{A}

We needed inequalities on correlation functions in different spin directions. Such inequalities go back at least to \cite{KK}. The present lemma appeared in this form in \cite{FU}. It also holds with periodic boundary conditions.

\begin{lemma}
\label{lem corr ineq}
Assume that, for all $x,y \in \Lambda$, the coupling constants satisfy
\[
|J_{x-y}\scr2| \leq J_{x-y}\scr1.
\]
Then we have that
\[
\bigl| \langle S_{0}\scr2 S_{x}\scr2 \rangle_{\Lambda,\beta,h} \bigr| \leq \langle S_{0}\scr1 S_{x}\scr1 \rangle_{\Lambda,\beta,h} ,
\]
for all $x \in \Lambda$.
\end{lemma}

\begin{proof}
Let $|a\rangle$, $a \in \{-S,\dots,S\}$ denote basis elements of $\bbC^{2S+1}$. Let the operators $S\scr{\pm}$ be defined by
\be
\begin{split}
&S\scr{+}|a\rangle = \sqrt{S(S+1) - a(a+1)} \; |a+1\rangle, \\
&S\scr{-}|a\rangle = \sqrt{S(S+1) - (a-1)a} \; |a-1\rangle,
\end{split}
\ee
with the understanding that $S\scr{+} |S\rangle = S\scr{-} |-S\rangle = 0$. Then let $S\scr1 = \frac12 (S\scr{+}+S\scr{-})$, $S\scr2 = \frac1{2\ii} (S\scr{+}-S\scr{-})$, and $S\scr3 |a\rangle = a |a\rangle$. It is well-known that these operators satisfy the spin commutation relations. Further, the matrix elements of $S\scr1, S\scr{\pm}$ are all nonnegative, and the matrix elements of $S\scr2$ are all less than or equal to those of $S\scr1$ in absolute values. Using the Trotter formula and multiple resolutions of the identity, we have
\be
\label{Trotter estimate}
\begin{split}
&\bigl| \Tr S_{0}\scr2 S_{x}\scr2 \e{-\beta H_{\Lambda,0}} \bigr| \leq \lim_{N\to\infty} \sum_{\sigma_{0}, \dots, \sigma_{N} \in \{-S,\dots,S\}^{\Lambda}} \biggl| \langle \sigma_{0} | S_{0}\scr2 S_{x}\scr2 | \sigma_{1} \rangle \\
&\langle \sigma_{1} | \e{\frac\beta N \sum J_{y-z}\scr3 S_{y}\scr3 S_{z}\scr3} | \sigma_{1} \rangle \langle \sigma_{1} | \Bigl( 1 \! + \! \tfrac\beta N \sum_{y,z \in \Lambda} ( J_{y-z}\scr1 S_{y}\scr1 S_{z}\scr1 + J_{y-z}\scr2 S_{y}\scr2 S_{z}\scr2 ) \Bigr) | \sigma_{2} \rangle \\
\dots &\langle \sigma_{N} | \e{\frac\beta N \sum J_{y-z}\scr3 S_{y}\scr3 S_{z}\scr3} \! | \sigma_{N} \rangle \langle \sigma_{N} | \Bigl( 1 \! + \! \tfrac\beta N \!\! \sum_{y,z \in \Lambda} \! ( J_{y-z}\scr1 S_{y}\scr1 S_{z}\scr1 \! + \! J_{y-z}\scr2 S_{y}\scr2 S_{z}\scr2 ) \Bigr) | \sigma_{0} \rangle \biggr|.
\end{split}
\ee
Observe that the matrix elements of all operators are nonnegative, except for $S_{0}\scr2 S_{x}\scr2$. Indeed, this follows from
\be
\begin{split}
&J_{y-z}\scr1 S_{y}\scr1 S_{z}\scr1 + J_{y-z}\scr2 S_{y}\scr2 S_{z}\scr2 \\
&= \tfrac14 (J_{y-z}\scr1 \! - \! J_{y-z}\scr2) (S_{y}\scr{+} S_{z}\scr{+} \! + \! S_{y}\scr{-}S_{z}\scr{-}) + \tfrac14 (J_{y-z}\scr1 \! + \! J_{y-z}\scr2) (S_{y}\scr{+} S_{z}\scr{-} \! + \! S_{y}\scr{-} S_{z}\scr{+}) .
\end{split}
\ee
We get an upper bound for the right side of \eqref{Trotter estimate} by replacing $|\langle \sigma_{0}| S_{0}\scr2 S_{x}\scr2 |\sigma_{1}\rangle|$ with $\langle\sigma_{0}| S_{0}\scr1 S_{x}\scr1 |\sigma_{1}\rangle$. We have obtained
\be
\bigl| \Tr S_{0}\scr2 S_{x}\scr2 \e{-\beta H_{\Lambda,0}} \bigr| \leq \Tr S_{0}\scr1 S_{x}\scr1 \e{-\beta H_{\Lambda,0}},
\ee
which proves the claim. We actually set $h=0$ in order to shorten the equations, but adding terms involving the $S_x\scr3$ operators to the Hamiltonian is straightforward.
\end{proof}


\begin{ack}
The authors are grateful to Robert Seiringer for encouragements and helpful comments on a draft. We also thank J\"urg Fr\"ohlich, Bruno Nachtergaele, Lorenzo Taggi, Balint T\'oth, and Yvan Velenik, for useful comments.
\end{ack}



\end{document}